\documentclass{article}
\usepackage{amsthm,authblk,fullpage,amssymb}

  \usepackage[noend, noline, ruled]{algorithm2e}
  \SetAlFnt{\normalsize}
  \DontPrintSemicolon
  \SetKwComment{Comment}{$\triangleright$\ }{}
  \usepackage[utf8]{inputenc}
  \usepackage{amsfonts,amsmath}
  \usepackage{tikz}
 \usetikzlibrary{decorations.pathreplacing,calc}
 \usepackage[hidelinks]{hyperref}
  \usepackage[capitalise]{cleveref}

  \newcommand{\floor}[1]{\left\lfloor #1 \right\rfloor}
  \newcommand{\ceil}[1]{\left\lceil #1 \right\rceil}
  \newcommand{\Oh}{\mathcal{O}}
  \newcommand{\polylog}{\mathrm{polylog\,}}
  \renewcommand{\S}{\mathbf{S}}
  \newcommand{\D}{\mathbf{D}}

  \newcommand{\Find}{\mathrm{Find}}
  \newcommand{\Union}{\mathrm{Union}}

  \newcommand{\LCE}{\mathrm{LCE}}
  \newcommand{\ShortLCE}{\mathrm{ShortLCE}}
  \newcommand{\SpecLCE}{\mathrm{SparseShortLCE}}
  \newcommand{\CoarseLCE}{\mathrm{CoarseLCE}}
  \newcommand{\code}{\mathrm{code}}
  \newcommand{\rank}{\mathrm{rank}}
  \newcommand{\unioncount}{\mathsf{\#union}}

  \theoremstyle{plain}
  \newtheorem{theorem}{Theorem}[section]
  \newtheorem{lemma}[theorem]{Lemma}  
  \newtheorem{corollary}[theorem]{Corollary}

  \crefname{figure}{Figure}{Figures}
  \theoremstyle{definition}

  \title{Faster Longest Common Extension Queries\\ in Strings over General Alphabets}

\author{Paweł Gawrychowski\footnote{Work done while the author held a post-doctoral position at Warsaw Center of Mathematics and Computer Science.}}
\author{Tomasz Kociumaka\footnote{Supported by Polish budget funds for science in 2013-2017 as a research project under the `Diamond Grant' program.}}
\author{Wojciech Rytter\footnote{Supported by the grant NCN2014/13/B/ST6/00770 of the Polish Science Center.}}
\author{Tomasz Waleń$^\ddagger$}
\affil{Institute of Informatics, University of Warsaw\\
    \texttt{[gawry,kociumaka,rytter,walen]@mimuw.edu.pl}}

\date{}

\begin{document}

\maketitle
  
\begin{abstract}
Longest common extension queries (often called longest common prefix queries) constitute
a~fundamental building block in multiple string algorithms, for example computing runs
and approximate pattern matching. 
We show that a sequence of $q$ LCE queries for
a~string of size $n$ over a general ordered alphabet can be realized in $\Oh(q \log \log n+n\log^*n)$ time making only $\Oh(q+n)$ symbol comparisons. 
Consequently, all runs in a string over a general
ordered alphabet can be computed in $\Oh(n \log \log n)$ time making $\Oh(n)$ symbol
comparisons. Our results improve upon a solution by Kosolobov (Information Processing Letters, 2016),
who gave an algorithm with $\Oh(n \log^{2/3} n)$ running time and conjectured that $\Oh(n)$ time is possible. 
We make a significant progress towards resolving this conjecture.
Our techniques extend to the case of general unordered alphabets, when the time increases to
$\Oh(q\log n + n\log^*n)$. The main tools are difference covers and the disjoint-sets data structure.
\end{abstract}

\section{Introduction}\label{sec:intro}
While many text algorithms are designed under the assumption of integer alphabet sortable in linear
time, in some cases it is enough to assume general alphabet. A general alphabet can be either
ordered, meaning that one can check if one symbol is less than another, or unordered, meaning
that only equality of two symbols can be checked.
Many classical linear-time string-matching
algorithms (e.g. Knuth-Morris-Pratt, Boyer-Moore) work for any unordered general alphabet.
Recently, a linear-time algorithm for computing the leftmost critical factorization in such model was
given~\cite{DBLP:journals/corr/Kosolobov15b}.
On the other hand, algorithms related to detecting repetitions usually need $\Omega(n\log n)$
equality tests~\cite{DBLP:journals/jal/MainL84}, and an on-line algorithm matching this bound
is known~\cite{DBLP:conf/cpm/Kosolobov15}.

In this paper we consider the longest common extension problem ($\LCE$, in short) in case of
general ordered and unordered alphabets.
The goal is to preprocess a given word $w$ of length $n$ for queries $\LCE(i,j)$
returning the length of the longest common factor starting at position $i$ and $j$
in $w$. Such queries are often a basic building block in more complicated algorithms,
for example in computing runs~\cite{B14, B14bis}
as well as in approximate string matching~\cite{DBLP:journals/jal/LandauV89}.

For integer alphabets of polynomial size, one can preprocess a given string in linear time and space to answer any $\LCE$ query in constant time.
Preprocessing space can be traded for query time~\cite{DBLP:conf/cpm/BilleGKLV15,DBLP:journals/jda/BilleGSV14} and
generalizations to trees~\cite{DBLP:conf/cpm/BilleGGLW15} and grammar-compressed
strings~\cite{DBLP:conf/stringology/Inenaga15,DBLP:journals/njc/KarpinskiRS97,DBLP:conf/cpm/Lifshits07,DBLP:conf/cpm/MiyazakiST97} are known.
The situation is more complicated for general alphabets. If the alphabet is ordered,
then of course we can reduce it to $[1..n]$ by sorting the
characters in $\Oh(n\log n)$ time and preprocess the obtained string in linear time and
space to answer any $\LCE$ query in constant time. However this increases the
total preprocessing time to $\Oh(n\log n)$. For unordered alphabet the situation is even
worse, because the reduction would take $\Oh(n^2)$ time.
A natural question is hence how efficiently we can answer a collection of such queries given one by one (on-line),
where we measure the preprocessing time plus the total time taken by all the queries.

It is known that if we can perform on-line $\Oh(n)$ $\LCE$ queries for
a given word of length $n$ in total time $T(n)$ making $\Oh(n)$ symbol comparisons,
then we can compute all runs in $\Oh(n+T(n))$ time making only $\Oh(n)$ symbol comparisons. 
An algorithm with $T(n)=\Oh(n\log^{2/3}n)$ time was recently presented by Kosolobov~\cite{Kosolobov2016241}, 
who posed the existence of a linear-time algorithm as an open question. 
Much earlier, Breslauer~\cite{BreslauerPHD} asked in his PhD thesis whether an easier task of square detection (equivalently,
checking if a word has at least one run) is possible in linear time in the comparison model.
In this paper we make a significant progress towards answering
both questions by giving a faster algorithm with $T(n) = \Oh(n\log\log n)$.

\subparagraph*{Our result} For a given string of length $n$ over a general ordered alphabet, we
can answer on-line a~sequence of $q$ LCE queries
in $\Oh(q\log\log n+n\log^*n)$ time making $\Oh(q+n)$ symbol comparisons.
In particular, a sequence of $\Oh(n)$ queries can be answered in $\Oh(n\log\log n)$ time.
Consequently, all runs in a string over a general ordered alphabet can be computed
in $\Oh(n \log\log n)$ time making $\Oh(n)$ symbol comparisons. 
For a general unordered alphabet we answer $q$ LCE queries in $\Oh(q\log n + n\log^* n)$ time,
still making $\Oh(q+n)$ symbol comparisons.

\subparagraph*{Overview of the methods} At a very high level, our approach is similar
to the one used by Kosolobov. We first show how to calculate $\min(\LCE(i,j),t)$ efficiently,
where $t=\polylog n$. Then we use a difference cover to sample some positions in the text.
Using ``short'' queries, we can efficiently construct a sparse suffix array for these sampled
positions, which in turn allows us to calculate an arbitrary $\LCE(i,j)$ efficiently. 
The key difference is that instead of calculating $\min(\LCE(i,j),t)$ naively, we use a recursive approach. 
The main tool there is an efficient Union-Find structure. 
This is enough to answer $\Oh(n)$ short queries in $\Oh(n\log\log n \cdot \alpha(n\log\log n, n\log\log n))$ total time. 
We can remove the $\alpha(n\log\log n, n\log\log n)$ factor introducing another difference cover and carefully analyzing the running
time of the Union-Find data structure.
Finally, we modify the algorithm to work faster when the number of queries $q$ is smaller than $n$. The main insight
allowing us to obtain $\Oh(q\log\log n+n\log^*n)$ total time is introducing multiple
levels of difference covers with some additional properties. 
Such family of difference covers was implicitly provided in~\cite{ourWADS}.


\section{Preliminaries}\label{sec:prelim}
\subsection{$t$-covers}
A difference cover is a number-theoretic tool used throughout the paper.
A set $\D\subseteq [0..t-1]$ is said to be a $t$-{\em difference-cover} if
$[0..t-1]\;=\; \{\, (x-y) \bmod t\;:\; x,y\in \D\,\}.$

\begin{lemma}[Maekawa~\cite{Maekawa}]\label{lem:difference-cover} For every integer $t$ there is
$t$-difference-cover of size $\Oh(\sqrt{t})$, which can be constructed in
$\Oh(\sqrt{t})$ time.
\end{lemma}

A subset $X$ of $[1..n]$ is $t$-\emph{periodic} if for each $i\in [1..n-t]$ we have: 
$i\in X\,\Leftrightarrow\, i+t\in X$.

A set $\S\subseteq [1..n]$ is called a $t$-\emph{cover} of $[1..n]$
if $\S$ is $t$-periodic and there is a constant-time computable function $h$ such that for $1\le i,j\le n-t$ we have 
$0\le h(i,j)\le t$ and $i+h(i,j), j+h(i,j)\in \S(t)$ (see \cref{fig:diff_cover_example}).

A $t$-cover can be obtained by taking a $t$-difference-cover $\D$ and
setting $\S(t) = \{i\in [1..n] : i \bmod t \in \D\}$. This is a well-known construction
implicitly used in~\cite{BurkhardtEtAl2003}, for example.

\begin{lemma}
For each $t\le n$ there is a $t$-cover $\S(t)$ of size  $\Oh(\frac{n}{\sqrt{t}})$ which can be constructed in $\Oh(\frac{n}{\sqrt{t}})$ time.
\end{lemma}

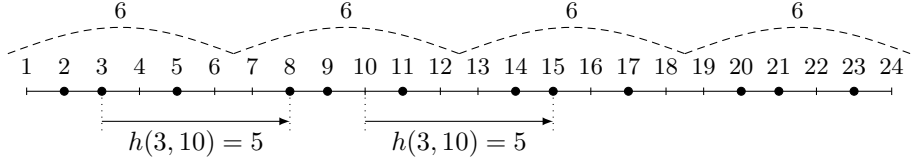
\begin{figure}[t]

\begin{center}
\begin{tikzpicture}[scale=0.5]
\tikzstyle{s} = [draw, circle, fill=black, minimum size = 3pt, inner sep = 0 pt, color=black]
\draw (1,0)--(24,0);
\foreach \i in {.5,6.5,...,18.5} {
	\draw[thin,densely dashed] (\i, 1) sin (\i+3, 1.7) node[above]{\small 6} cos (\i+6,1);
}
 
\foreach \i in {1,...,24} {
  \draw (\i,-0.1)--+(0,0.2);
  \node at (\i, 0.2) [above] {\small \i};
}
\foreach \i in {2,3,5,8,9,11,14,15,17,20,21,23} {
  \node [s] at (\i,0) {};
}

\foreach \i/\l in {3/5, 10/5} {
  \draw [dotted] (\i,0)--+(0,-1.1);
  \draw [dotted] (\i+\l,0)--+(0,-1.1);
  \draw [-latex] (\i,-.8)--node[midway,below] {$h(3,10)=5$} +(\l,0);
}
\end{tikzpicture}
\end{center}

\caption{An example of a 6-cover $\S(6)=\{2,3,5, 8,9,11, 14,15,19, 20,21,23\}$ (for $\D=\{2,3,5\}$), with the
elements marked as black circles. For example, we have $h(3,10)=5$, since $3+5,\,10+5\in \S(6)$.}
\label{fig:diff_cover_example}
\end{figure}

\subsection{Disjoint-sets structure}

Our another tool is a disjoint-sets data structure.
In this problem we maintain a family of disjoint subsets of $[1..n]$, initially consisting of
singleton sets. We perform $\Find$ queries asking for a subset containing a given element,
and $\Union$ operations which merge two subsets.

Note that the extremely fast-growing Ackermann function~\cite{DBLP:journals/jacm/TarjanL84} is defined for $i,j \in \mathbb{Z}_{> 0}$ as
$$A(i,j)=\begin{cases}2^j & \text{if }i=1,\\A(i-1,2) &\text{if }i>1\text{ and }j=1,\\A(i-1, A(i,j-1))&\text{if }i>1\text{ and }j>1.\end{cases}$$
Moreover, for $n,m\in \mathbb{Z}_{>0}$ ($m\ge n$) one defines $\alpha(m,n) = \min \{ i \ge 1 : A(i, \floor{\tfrac{m}{n}}) > \log n\}$.

\begin{lemma}[Tarjan~\cite{UnionFind1975}]\label{lem:uf}
A sequence of up to $n$ $\Union$ and $m$ $\Find$ operations
on an $n$-element set can be executed on-line in $\Oh(n+m\cdot \alpha(m+n,n))$ total time.
\end{lemma}

\begin{lemma}\label{lem:ackermann}
For every $n,m\in \mathbb{Z}_{>0}$, we have $n+m\cdot \alpha(m+n,n)=\Oh(m+n\log^* n)$.
\end{lemma}
\begin{proof}
First, observe that the Ackermann function $A(i,j)$ is monotone with respect to both coordinates and that $A(i,j)\ge 2^{i+j-1}$. These properties are easy to show by induction.
Additionally, let us recall the row inverse of the Ackermann function is defined for $i,n\in \mathbb{Z}_{>0}$ as $a(i,n)=\min\{ j \ge 1 : A(i,j) > \log n\}$.
Note that $\alpha(m,n)\le i$ if $\frac{m}{n}\ge a(i,n)$, in particular, $\alpha(m,n)\le 2$ if $m\le n\log^* n$; see~\cite{DBLP:journals/jacm/TarjanL84}.
We shall make two claims relating the $\alpha(m,n)$ and $a(i,n)$ functions.
First, $\alpha(n,n)\le 4+a(3,n)$ for every $n\in \mathbb{Z}_{>0}$. This follows from:
\begin{multline*}
A(4+a(3,n),1)=A(3+a(3,n),2)=A(2+a(3,n),A(3+a(3,n),1))\ge\\  A(3,2^{3+a(3,n)})\ge A(3,a(3,n)) > \log n.
\end{multline*}
Moreover, the fact that  $a(3,n)\le 2+\ceil{\log a(2,n)}$ for every $n\in \mathbb{Z}_{>0}$ is a consequence of:
\begin{multline*}
A(3,2+\ceil{\log a(2,n)})=A(2, A(3,1+\ceil{\log a(2,n)}))\ge  A(2, 2^{3+\ceil{\log a(2,n)}})\ge\\  A(2,8a(2,n)) \ge A(2,a(2,n))>\log n.
\end{multline*}

To prove the lemma, we consider two cases: If $m+n \ge n\cdot a(3,n)$, then
$m\cdot \alpha(m+n,n)\le 3m$, so $\Oh(n+m\cdot \alpha(m+n,n))=\Oh(n+m)$.
Otherwise, the claims that we made above imply:
$$m\cdot \alpha(m+n,n)\le m\cdot \alpha(n,n)\le n\cdot a(3,n)\cdot (4+a(3,n))\le n(2+\ceil{\log a(2,n)})(6+\ceil{\log a(2,n)}),$$
i.e., $n+m\cdot \alpha(m+n,n)=\Oh(n\log^2 a(2,n))=\Oh(n \cdot a(2,n))=\Oh(n\log^* n)$,
as desired.
\end{proof}

\section{Generic LCE algorithm for general ordered alphabets}

We define $t$-short $\LCE$ queries by restricting the answer to at most $t$:
$$\ShortLCE_t(i,j)\;=\; \min(\LCE(i,j),\,t).$$
We define a $t$-block as a fragment of the input text $w$ which starts in $\S(t)$ and has length~$t$. 
If a position in $\S(t)$ lies near the end of $w$, we form a $t$-block from a suffix of $w$ and enough dummy symbols to reach length $t$.
We also introduce $t$-coarse LCE queries, which are LCE queries
restricted to positions from $\S(t)$ returning the number of matching $t$-blocks:
$$ \CoarseLCE_t(i,j)\;=\;
  \begin{cases}
   \lfloor \LCE(i,j)/t \rfloor  & \quad \text{if } i,j\in \S(t),\\
     \bot& \quad \text{otherwise.} \\
  \end{cases}
$$

We now describe how to use $\ShortLCE$ and $\CoarseLCE$ queries for general $\LCE$ queries.
\begin{lemma} \label{lem:generic}
If every sequence of $q$ $\ShortLCE_t$ queries and $\CoarseLCE_t$ queries can
be executed on-line in total time $T(n,q)$,
 then every sequence of $q$ $\LCE$ queries can be executed on-line in total time $T(n,\Oh(q))+\Oh(n+q)$.
\end{lemma}

\begin{proof}
To calculate $\LCE(i,j)$ we first check if $\LCE(i,j) < t$ by calling
$\ShortLCE_t(i,j)$. If so, we are done. Otherwise, we can reduce computing
$\LCE(i,j)$ to computing $\LCE(i+\Delta,j+\Delta)$ for any $\Delta \leq t$.
In particular, we can choose $\Delta = h_t(i,j)$ so that $i+\Delta, j+\Delta \in \S(t)$.
Then we call $\CoarseLCE_t(i+\Delta,j+\Delta)$ which gives us the value $\lfloor \frac1t(\LCE(i,j)-\Delta) \rfloor$.
Computing the exact value of $\LCE(i,j)$ requires another $\ShortLCE_t$ query; see \cref{alg:lce_generic}.
The whole process is illustrated in \cref{fig:generic}.
\end{proof}

\begin{algorithm}[h]
\caption{$\mathrm{Generic}\LCE(i,j)$\label{alg:lce_generic}}
$\ell_1\, =\, \ShortLCE_t(i,j)$ \;
\lIf{$\ell_1 < t$}{%
  \KwRet{$\ell_1$}
}
$\Delta\, =\, h_t(i,j)$\Comment*[r]{$i+\Delta,j+\Delta\in \S(t)$}
  $\ell_2$\, =\, $t\cdot \CoarseLCE_t(i+\Delta,j+\Delta)$ \;
  $\ell_3$\, =\, $\ShortLCE_t(i+\Delta+\ell_2, j+\Delta+\ell_2)$ \;
\vskip 0.1cm
  \KwRet{$\Delta+\ell_2+\ell_3$} \;
\end{algorithm}

\begin{figure}[h]
\centering
\begin{tikzpicture}[xscale=0.5,yscale=0.3]
\def\lenT{2}
\def\lenD{1.2}
\def\lenLA{\lenT}
\def\lenLB{3*\lenT}
\def\lenLC{1.5}
\def\lenS{12}
\newcommand\period[2]{
  \foreach \i in {0,1,2} {
    \draw (#1+\i*\lenT, #2) sin (#1+\i*\lenT+0.5*\lenT, #2+0.5) cos (#1+\i*\lenT+\lenT, #2);
  }
  \node at (#1+0.5*\lenT, #2+0.5) [above] {\small $t$};
}

\tikzstyle{s} = [draw, circle, fill=black, minimum size = 3pt, inner sep = 0 pt, color=black]

\draw (0,0)--+(10,0);
\draw (\lenS,0)--+(10,0);

\foreach \x in {0,\lenD,\lenD+\lenLB} {
  \node[s] at (\x,0) {};
  \node[s] at (\lenS+\x,0) {};
}
\node at (0,0) [above left] {$i$};
\node at (\lenS,0) [above left] {$j$};

\foreach \i/\l/\t/\h in {
  0/\lenLA/$\ell_1$/-4,
  \lenD/\lenLB/$\ell_2$/-2,
  \lenD+\lenLB/\lenLC/$\ell_3$/-4
} {
  \draw [-latex] (\i,\h)--node[midway,below] {\t} +(\l,0);
  \draw [-latex] (\lenS+\i,\h)--node[midway,below] {\t} +(\l,0);
}

\draw [latex-latex] (0,1)--node[midway, above] {\small $\Delta$} +(\lenD,0);
\draw [latex-latex] (\lenS,1)--node[midway, above] {\small $\Delta$} +(\lenD,0);

\foreach \x/\y in {0/-4, \lenD/-2, \lenD+\lenLB/-4, \lenD+\lenLB+\lenLC/-4, 0/1, \lenD/1} {
  \draw [dotted] (\x,0)--(\x,\y);
  \draw [dotted] (\lenS+\x,0)--(\lenS+\x,\y);
}

\period{\lenD}{-2}
\period{\lenS+\lenD}{-2}

\node at (\lenS+10,-2) [right] {$\CoarseLCE$};
\node at (\lenS+10,-4) [right] {$\ShortLCE$};

\end{tikzpicture}
\caption{Illustration of \cref{alg:lce_generic} for the case $\ell_1 \ge \Delta$.}
\label{fig:generic}
\end{figure}
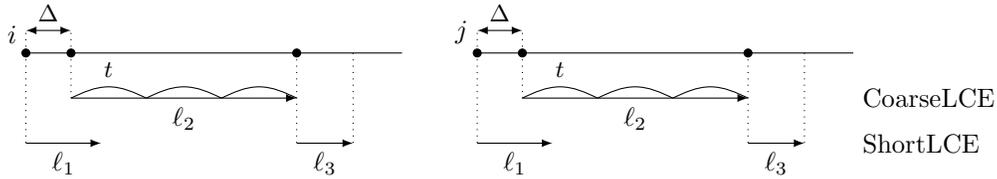

\section{$\ShortLCE_t$ queries in $\Oh(\log t)$ amortized time}
In this section we show how to implement fast on-line $\ShortLCE_t$ queries.
We assume that $t=2^k$ and set $t'=\Theta(\log t)$ to be a smaller power of two. 
The amortized running time is $\Oh(\log t + \sqrt{\log t}\log^* n)$,
which in particular is $\Oh(\log t)$ for $t=\log^{\Omega(1)} n$.
The key components are Union-Find structures and $t'$-covers. 
We start with a simpler (and slightly slower) algorithm without $t'$-covers.

\subsection{$\ShortLCE_t$ queries in  $\Oh(\log t \cdot \alpha((n+q)\log t, n\log t))$ amortized time}

\begin{lemma}
\label{lem:recursive}
A sequence of $q$ $\ShortLCE_{2^k}(i,j)$  queries can be executed on-line in
total time $\Oh((q+n)k\cdot \alpha((q+n)k,nk))$.
\end{lemma}
\begin{proof}
We compute $\ShortLCE_{2^k}(i,j)$ using a recursive procedure; see \cref{alg:short}.
The procedure first checks if $w[i..i+2^k-1]$ is already known to be equal to $w[j..j+2^k-1]$ using a Union-Find structure. 
If so, we are done. 
Otherwise, if $k=0$, we simply compare $w[i]$ and $w[j]$.
If $k>0$, we recursively calculate $\ShortLCE_{2^{k-1}}(i,j)$ and, if the call returns $2^{k-1}$, also $\ShortLCE_{2^{k-1}}(i,j)$.
Finally, if both calls return $2^{k-1}$, we update the Union-Find structure to store that
$w[i..i+2^k-1]=w[j..j+2^k-1]$.

\begin{algorithm}[ht]
\caption{$\ShortLCE_{2^k}(i,j)$: compute $\LCE(i,j)$ up to length $2^k$\label{alg:short}}
\lIf{$\Find_k(i)=\Find_k(j)$}{%
 \KwRet{$2^k$}
}
\vskip 0.15cm
\eIf{$k=0$}{
  \leIf{$w[i]=w[j]$}{$\ell=1$}{$\ell=0$}
}{
 $\ell = \ShortLCE_{2^{k-1}}(i,j)$ \;
 \If{$\ell=2^{k-1}$}{
   $\ell = 2^{k-1} + \ShortLCE_{2^{k-1}}(i+2^{k-1},j+2^{k-1})$ \;
  }
}
\vskip 0.15cm
\lIf{$\ell=2^k$}{%
  $\Union_k(i,j)$
}
\vskip 0.15cm
\KwRet{$\ell$}
\vskip 0.1cm
\end{algorithm}
To analyze the complexity of the procedure, we first observe that the total number of
calls to $\Union$ is $\Oh(nk)$, because each such call discovers that $w[i..i+2^k-1]=w[j..j+2^k-1]$ (which was not known before). 
Moreover, these calls contribute $\Oh(nk)$ to the total running time.
We argue that the number of executed $\Find$ queries and the running time of the remaining
operations performed by $\ShortLCE_{2^k}(i,j)$ is proportional to  $\Oh(k+1)$ plus the number of  $\Union$ calls, which implies the lemma.
For the sake of conciseness, $\unioncount$ denotes
the number of calls to $\Union$ triggered by the considered call to $\ShortLCE$ (including itself).

We inductively bound the number of recursive calls triggered by $\ShortLCE_{2^k}(i,j)$:
\begin{align*}
&2k+1+2\unioncount && \text{if }w[i..i+2^k-1]\neq w[j..j+2^k-1],\\
&1+2\unioncount && \text{if }w[i..i+2^k-1]= w[j..j+2^k-1].\\
\end{align*}
$\ShortLCE_{1}$ terminates immediately, so this holds for $k=0$. For $k>0$ we have four cases.
\begin{enumerate}
\item $w[i..i+2^k-1]$ is already known to be equal to $w[j..j+2^k-1]$. Then we terminate immediately.
\item $w[i..i+2^{k-1}-1] \neq w[j..j+2^{k-1}-1]$. Then the number of recursive calls triggered
by $\ShortLCE_{2^{k-1}}(i,j)$ is $2k-1+2\unioncount$  so the number of recursive
calls triggered by $\ShortLCE_{2^k}(i,j)$ is $2k+2\unioncount$.
\item $w[i..i+2^{k-1}-1] = w[j..j+2^{k-1}-1]$ but $w[i+2^{k-1}..i+2^{k}-1] \neq w[j+2^{k-1}..j+2^{k}-1]$.
The number of recursive calls triggered by $\ShortLCE_{2^{k-1}}(i,j)$ and
$\ShortLCE_{2^{k-1}}(i+2^{k-1},j+2^{k-1})$ is $1+2\unioncount$
and $2k-1+2\unioncount$, respectively.
The total number of triggered recursive calls is hence $2k+1+2\unioncount$.
\item $w[i..i+2^{k-1}-1] = w[j..j+2^{k-1}-1]$ and $w[i+2^{k-1}..i+2^{k}-1] = w[j+2^{k-1}..j+2^{k}-1]$.
The number of recursive calls triggered by both $\ShortLCE_{2^{k-1}}(i,j)$ and
$\ShortLCE_{2^{k-1}}(i+2^{k-1},j+2^{k-1})$ is $1+2\unioncount$.
However, $w[i..i+2^{k}-1]$ was not known to be equal to $w[j..j+2^{k}-1]$, so we then execute $\Union_k(i,j)$.
Hence the total number of recursive calls is $1+2\unioncount$ (rather than of $3+2\unioncount$).
\end{enumerate}

Consequently, the total running time follows from \cref{lem:uf}.
\end{proof}

\subsection{Faster $\ShortLCE_t$ queries}
Assume $t=2^k=\Omega(\log n)$.
We show how to reduce the factor $\alpha(qk+nk,nk)$ introducing a $t'$-cover,
for $t'=2^{k'}$.
We define a sparse version of $\ShortLCE$ queries, which are $\ShortLCE$ queries restricted to positions from $\S(t')$:
\[ \SpecLCE_{t,t'}(i,j)\;=\;
  \begin{cases}
   \ShortLCE_t(i,j)  & \quad \text{if } i,j\in \S(t')\\
    \bot & \quad \text{otherwise } \\
  \end{cases}
\]
We slightly modify \cref{alg:short} to obtain \cref{alg:spec}, which computes
$\min(\LCE(i,j),2^k)$ for positions $i,j\in \S(t')$.

\begin{algorithm}[b]
\caption{$\SpecLCE_{2^k,2^{k'}}(i,j)$: compute $\min(\LCE(i,j),2^k)$ for $i,j\in \S(2^{k'})$}\label{alg:spec}
\lIf{$\Find_k(i)=\Find_k(j)$}{%
 \KwRet{$2^k$}
}\vskip 0.15cm
\If{$k=k'$}{Compute naively $\ell = \ShortLCE_{2^{k'}}(i,j)$\;}
\Else{
 $\ell = \SpecLCE_{2^{k-1},2^{k'}}(i,j)$ \;
 \If{$\ell=2^{k-1}$}{
   $\ell = 2^{k-1} + \SpecLCE_{2^{k-1},2^{k'}}(i+2^{k-1},j+2^{k-1})$ \;
  }
}
\vskip 0.15cm
\lIf{$\ell=2^k$}{%
  $\Union_k(i,j)$ 
}
\vskip 0.15cm
\KwRet{$\ell$}
\vskip 0.1cm
\end{algorithm}
\begin{figure}[t]
\centering
\scalebox{0.7}{\begin{tikzpicture}
\tikzstyle{level 1}=[level distance=1.5cm, sibling distance=5cm]
\tikzstyle{level 2}=[level distance=1.5cm, sibling distance=2cm]
\tikzstyle{level 3}=[level distance=1.5cm, sibling distance=2cm]
\tikzstyle{level 4}=[level distance=1.5cm, sibling distance=1.5cm]
\node {Find}
  child {
     node {Union}
     child {
        node {Union}
        child { node {Find} }
        child {
           node {Union}
           child { node {Find} }
           child {
              node[text width=7em, text centered] {Naive and Union}
           }
        }
     }
     child {
        node {Find}
     }
  }
  child {
     node {Find}
     child { 
        node { Find }
        child {
           node { Union }
           child {
              node[text width=7em, text centered] {Naive and Union}
           }
           child { node {Find} }
        }
        child { 
           node { Find } 
           child { node { Naive } } 
        } 
     } 
  }
;
\end{tikzpicture}
}
\caption{A recursion tree of $\SpecLCE_{t,t'}(i,j)$
for some example parameters  such that $t=2^4t'$.
The calls terminating with $\Union$, $\Find$ and
naive tests (in a segment of size $t'$) are shown as nodes in the figure. 
The naive tests are only at the bottom of
the tree and they are accompanied by Unions (except the last one).  }
\end{figure}
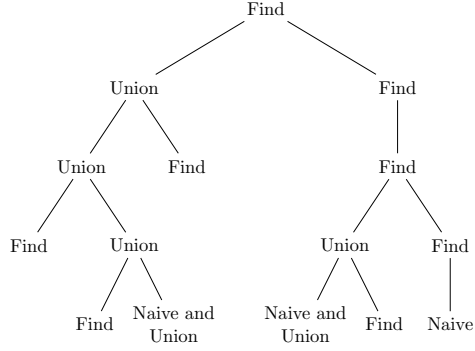

\begin{lemma}
\label{lem:recursive2}
A sequence of $q$ $\SpecLCE_{2^k,2^{k'}}$ queries can be executed
on-line in total time $\Oh(q(k+2^{k'})+n\sqrt{2^{k'}} + \frac{nk}{\sqrt{2^{k'}}}\log^*n)$.
\end{lemma}
\begin{proof}
The analysis is similar to the proof of \cref{lem:recursive}.
The total number of calls to $\Union$ is now only $\Oh(\frac{nk}{2^{k'/2}})$
because we always have that $i,j\in \S(2^{k'})$. 
Hence, excluding the cost of computing $\ell = \ShortLCE_{2^{k'}}(i,j)$,
 the total time complexity is $\Oh(qk + \frac{nk}{2^{k'/2}}\log^*n)$ by the same reasoning as in \cref{lem:recursive},
except that we additionally apply \cref{lem:ackermann} to bound the running time of the Union-Find data structure (stated in \cref{lem:uf}).

Now we analyze the cost of computing $\ell = \ShortLCE_{2^{k'}}(i,j)$.
First, observe that for every original call to $\SpecLCE_{2^k,2^{k'}}(i,j)$ we have
at most one such computation with $\ell < 2^{k'}$ (because it means that we have
found a mismatch and no further recursive calls are necessary).
On the other hand, if $\ell = 2^{k'}$, then we call $\Union_{k'}(i,j)$, 
which may happen at most $\frac{n}{2^{k'/2}}$ times. 
Therefore, the total complexity of all these naive computations is $\Oh(n2^{k'/2}+q\cdot 2^{k'})$.
\end{proof}

   \begin{algorithm}[t]
\caption{Faster$\ShortLCE_{2^k,2^{k'}}(i,j)$\label{alg:lce_version_2}}
    Compute naively $\ell = \ShortLCE_{2^{k'}}(i,j)$\;
    \vskip 0.15cm
    \lIf{$\ell < 2^{k'}$}{\KwRet{$l$}}
    \vskip 0.15cm
    $\Delta = h_{2^{k'}}(i,j)$ \;
    $\ell = \Delta + \SpecLCE_{2^{k},2^{k'}}(i+\Delta, j+\Delta)$\;
    \KwRet{$\min(\ell,2^k)$}
    \vskip 0.1cm
    \end{algorithm}

The next lemma is a direct consequence of \cref{lem:recursive2} and \cref{alg:lce_version_2}
with $2^{k'}=\Theta(k)$.

\begin{lemma}\label{lem:short}
A sequence of $q$ $\ShortLCE_{{2^k}}$ queries can be executed
on-line in total time $\Oh(q k + n\sqrt{k}\log^* n)$.
\end{lemma}

\section{$\CoarseLCE_t$ queries}
Let $t=\Omega(\log^2 n)$. Recall that we defined a  $t$-block of $w$ as a factor of size $t$ starting in $\S(t)$.
We want to show how to preprocess $w$ in $\Oh(n\log\log n)$ time,
so that any $\CoarseLCE_t$ query can be answered in constant time.
To this end we proceed as follows:
\begin{enumerate}
\item sort all $t$-blocks in lexicographic order and remove duplicates,
\item encode every $t$-block with its rank on the sorted list,
\item construct a new string $\code(w)$ of length $\Oh(n)$ over alphabet $[1..n]$,
such that any $\CoarseLCE_t$ query can be reduced to an LCE query on $\code(w)$,
\item preprocess $\code(w)$ for LCE queries.
\end{enumerate}

\begin{lemma}
\label{lem:rename}
For $t=\Omega(\log^2 n)$ we can lexicographically sort all $t$-blocks of $w$ in $\Oh(n\log t)$ time.
\end{lemma}
\begin{proof}
Two $t$-blocks can be lexicographically compared with a $\ShortLCE_{t}$ query.
We have $\Oh(\frac{n}{\sqrt{t}})$ such blocks, hence one of the classical
sorting algorithms they can be all sorted using $\Oh(\frac{n}{\sqrt{t}}\log n)=\Oh(n)$ queries.
By \cref{lem:short}, the total time to execute these queries and sort all $t$-blocks is therefore $\Oh(n\log t)$.
\end{proof}

We can use the lexicographic order of $t$-blocks to assign ranks to all $t$-blocks. 
Then we reduce $\CoarseLCE$ queries to LCE queries in a word $\code(w)$ over an integer alphabet; see \cref{fig:coarse}.
\begin{lemma}\label{lem:coarse}
For $t=\Omega(\log^2 n)$ we can preprocess $w$ in $\Oh(n\log t)$ time so that any 
$\CoarseLCE_t$ query can be answered in constant time.
\end{lemma}
\begin{proof}
Using \cref{lem:rename}, we assign a number to each $t$-block,
so that two $t$-blocks are identical if and only if their numbers are equal.
The number assigned to the block starting at position $p \in \S(t)$ is denoted
$\rank(p)$. These numbers are ranks on a sorted list of length $|\S(t)|$, so
$\rank(p) \in [1..|\S(t)|]$. Then we construct a new string $\code(w)$ as follows.
Let
$\{\,i_1,i_2,\ldots i_k\,\} \,=\, [1,t]\cap \S(t)$
and $z_s$ be the word obtained from $w$ by concatenating the numbers assigned
to all $t$-blocks starting at positions $i_s,i_s+t,i_s+2t,i_s+3t,\ldots$:
$$ z_s = \rank(i_s) \rank(i_s + t) \rank(i_s + 2t) \rank(i_s + 3t ) \ldots .$$
Finally, we introduce $k$ new distinct letters $\#_1,\#_2,\ldots,\#_s$ and
construct $\code(w)$:
$$code(w)\;=\; z_1\cdot\#_1\cdot z_2\cdot\#_2\cdot z_3\cdot\#_3 \cdots z_k\cdot\#_k.$$
Next, $\code(w)$ is preprocessed to answer $\LCE$ queries in constant time.
A $\CoarseLCE_t(p,q)$ query for positions $p,q \in \S(t)$ is answered by first computing positions $p',q'$ corresponding to $p,q$ in $\code(w)$. 
Formally, if $p = i_s \bmod t$, then $p' = |z_1 \#_1 z_2 \#_2 \ldots z_{s-1} \#_{s-1}| + \frac{p-i_s}{t}+1$;
$q'$ is computed similarly. Then an $\LCE(p',q')$ query on $\code(w)$ returns
$\CoarseLCE_t(p,q)$. 
The positions $p'$ and $q'$ can be computed in constant time, so the total query
time is constant. Preprocessing $\code(w)$ requires constructing its suffix array,
which takes linear time for integer alphabets of polynomial size,
and preprocessing it for range minimum queries, which also takes linear time.
Hence the total preprocessing time is $\Oh(n\log t)$.
\end{proof}

\begin{figure}[t]
\begin{center}
\scalebox{0.8}{\begin{tikzpicture}[scale=0.35]
\draw (1,0)--(25,0)--(25,1.3)--(1,1.3)--cycle;

\foreach \i in {2, 3, 6, 7, 10, 11, 12, 15, 16, 19, 20, 21} {
  \node at (\i+0.5,0) [above] {\tt a};
}

\foreach \i in {1, 4, 5, 8, 9, 13, 14, 17, 18, 22, 23, 24} {
  \node at (\i+0.5,0) [above] {\tt b};
}

\foreach \i in {25,26,27,28} {
  \node at (\i+0.5,0) [above] {\tt *};
}

\foreach \i in {2,3,5,8,9,11,14,15,17,20,21,23} {
  \draw[thick] (\i,0)--+(0,1.3);
  \node at (\i,1.3) [above] {\footnotesize \i};
}

\foreach \i/\l/\h/\t in {
  2/6/-2/1, 8/6/-2/8,  14/6/-2/6, 20/6/-2/2, 
  3/6/-4/3, 9/6/-4/5,  15/6/-4/1, 21/6/-4/4,
  5/6/-6/6, 11/6/-6/1, 17/6/-6/8, 23/6/-6/7
} {
  \draw (\i,\h) rectangle (\i+\l,\h+1);
  \node at (\i+0.5*\l,\h+0.5) {\small \tt \t};
}

\node at (-0.5,0.5) [left] {$w:$};
\node at (0,-1.5) {$\alpha$};
\node at (0,-3.5) {$\beta$};
\node at (0,-5.5) {$\gamma$};

\node at (-0.5,-7.5) [left] {$\code(w):$};

\draw (1.5,-8)--+(15,0)--+(15,1)--+(0,1)--cycle;
\foreach \x in {4,5,9,10,14}  {
  \draw (1.5+\x,-8)--+(0,1);
}

\foreach[count=\i] \x in {1,8,6,2,\$,3,5,1,4,\#,6,1,8,7,\&} {
  \node at (\i+1,-7.5) {\tt \x};
}

\foreach \x/\l/\t in {1/4/\alpha, 6/4/\beta, 11/4/\gamma} {
  \draw [decorate,decoration={brace,mirror,amplitude=4pt}]
              (\x+0.5,-8.2) --  node [midway,yshift=-0.4cm] {$\t$} +(\l,0);
}
\end{tikzpicture}}
\end{center}
\caption{$6$-blocks of $w$ are lexicographically sorted (using $\ShortLCE_t$) and ranked. 
Then $\CoarseLCE_6(2,11)$ in $w$ is reduced to  $\LCE(1,12)$ in $\code(w)$.\label{fig:coarse} }
\end{figure}
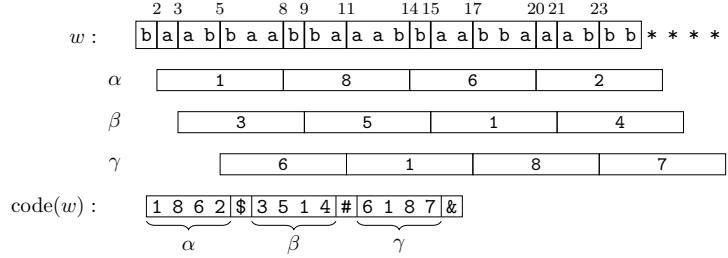

\begin{theorem}
A sequence of $\Oh(n)$ LCE queries for a string over a general ordered alphabet can be executed on-line in total time
$\Oh(n\log \log n)$ making only $\Oh(n)$ symbol comparisons.
\end{theorem}
\begin{proof}

We set $t=\Theta(\log ^2 n)$ and reduce each LCE query to constant number of
$\CoarseLCE_t$ queries and $\ShortLCE_t$ queries as described in \cref{lem:generic}.
Thus together with \cref{lem:short} and \cref{lem:coarse} we obtain that
any sequence of $q$ $\LCE$ queries for a string over a general ordered alphabet
can be realized in $\Oh(n\log\log n)$ time. However, the total number of symbol
comparisons used by the algorithm might be $\Omega(n\log\log n)$. This can
be decreased to $\Oh(n)$ with yet another Union-Find data structure, where we
maintain sets of positions already known to store the same letter. 
This is essentially the idea used in Lemma 7 of~\cite{DBLP:conf/stacs/Kosolobov15}.
\end{proof}

\section{Faster solution for sublinear number of queries}

The algorithm presented in the previous section is not efficient when the number
of queries $q$ is significantly smaller than the length of the string $n$. 
In this section we show that this can be avoided, and we present an $\Oh(q\log\log n+n\log^*n)$-time algorithm.
This requires some nontrivial changes in our approach. 
In particular, we need a stronger notion of $t$-covers, which form a \emph{monotone family}.

$\S(4^0),\S(4^1),\S(4^2),\ldots\subseteq [1,n]$ is a monotone family of covers if the
following conditions hold for every $k$:
\begin{enumerate}
\item $\S(4^k)$ is a $4^k$-cover (except that $h_{4^k}$ is computable in $\Oh(k)$ instead of constant time).
\item $\S(4^{k+1}) \subseteq \S(4^{k})$.
\item For any $i,j\in\S(4^k)$ we have that $h_{4^{k+1}}(i,j)\in\{0,4^k,2\cdot 4^k\}$, and furthermore
for such arguments $h_{4^{k+1}}$ can be evaluated in constant time.
\item $|\S(4^k)| \leq (\frac{3}{4})^k n$.
\end{enumerate}

The existence of such a family is not completely trivial, in particular plugging in the standard
construction of $\S(4^k)$ from \cref{lem:difference-cover} does not guarantee that
$\S(4^{k+1}) \subseteq \S(4^{k})$.
The following lemma, implicitly shown in \cite{ourWADS}, provides an efficient construction.

\begin{lemma}[Gawrychowski et al.~\cite{ourWADS}, Section~4.1]
Let $\S(4^k)$ be the set of non-negative integers $i\in [1,n]$ such that none of the
$k$ least significant digits of the base-$4$ representation of $i$ is zero.
Then $\S(4^0),\S(4^1),\S(4^2),\ldots$ is a monotone family of covers, which can be
constructed in $\Oh(n)$ total time.
\end{lemma}

\subsection{$\ShortLCE_t$ queries with monotone family of covers}\label{sec:faster}

Similarly as in the proof of \cref{lem:short}, we reduce $\ShortLCE$ queries to $\SpecLCE$ queries.
However, now we slightly change the definition of $\SpecLCE$ queries so that there is only one parameter as follows:
$$ \SpecLCE_{t}(i,j)\;=\;
  \begin{cases}
   \ShortLCE_t(i,j)  & \quad \text{if } i,j\in \S(t)\\
    \bot & \quad \text{otherwise } \\
  \end{cases}
$$

\begin{lemma}
\label{lem:spec2}
Consider a sequence of $q$ $\SpecLCE_{4^{k_i}}$ queries for $i\in \{1,\ldots,q\}$. 
The queries can be answered online in $\Oh((n + s)\cdot\alpha(n + s,n))$ time
where $s=\sum_{i=1}^q T_i$ with $T_i = 1$ if the $i$-th query returns $4^{k_i}$ and $T_i = k_i+1$ otherwise.
\end{lemma}

\begin{proof}
We maintain a separate $\Union$-$\Find$ structure for $\S(4^k)$ at every level $k\in  \{0,\ldots,K\}$ where $K=\max_{i=1}^q k_i$.
To answer a query for $\SpecLCE_{4^{k}}$, we check if $\Find_{k}(i)=\Find_{k}(j)$ and if so, return $4^{k}$. 
Otherwise, we calculate the answer with at most four calls to $\SpecLCE_{4^{{k-1}}}$.
This is possible because $\S(4^{k}) \subseteq \S(4^{k-1})$ and $\S(4^{k-1})$ is $4^{k-1}$-periodic.
Finally, we call $\Union_{k}(i,j)$ if the answer is $4^{k}$; see \cref{alg:spec_lce_4k}.

\begin{algorithm}[h]
    \caption{$\SpecLCE_{4^k}(i,j)$: compute $\min(\LCE(i,j),4^k)$ for $i,j\in \S(4^k)$\label{alg:spec_lce_4k}}
    \lIf{$\Find_k(i)=\Find_k(j)$}{%
     \KwRet{$4^k$}
    }
    \vskip 0.15cm
    \eIf{$k=0$}{
      \leIf{$w[i]=w[j]$}{$\ell=1$}{$\ell=0$}
    }{
     $\ell = 0$ \;
     \For{$p=0$ \KwSty{to} $3$}{
       $\ell = \ell + \SpecLCE_{4^{k-1}}(i+p\cdot 4^{k-1},j+p\cdot 4^{k-1})$ \;
       \lIf{$\ell < (p+1) \cdot 4^{k-1}$}{%
         \KwSty{break}
       }
     }
    }
    \vskip 0.15cm
    \lIf{$\ell=4^k$}{%
      $\Union_k(i,j)$
    }
    \vskip 0.15cm
    \KwRet{$\ell$}
    \vskip 0.1cm
    \end{algorithm}

We again analyze the number of recursive calls to $\SpecLCE_{4^{k}}$ counting $\Union$ operations.
The total number of unions at level $k$ is $|\S(4^k)|\le (\frac34)^k$, and in total this sums up to $\Oh(n)$.
The amortized number of $\Find$ queries executed by a call to $\SpecLCE_{4^{k}}$ is constant if $\LCE(i,j)=4^{k}$ and $\Oh(k+1)$ otherwise.
These values also bound the running time of the remaining operations.
Hence, by \cref{lem:uf}, the total time is as claimed.
\end{proof}

    \begin{figure}[ht]
    \centering
    \renewcommand{\arraystretch}{1.2}
    \begin{tabular}{cl}
    $k'$ & $\SpecLCE$ calls \\ \hline
    0 & $\SpecLCE_{4^0}(1013\mathbf{0}_4, 00101_4) \rightarrow \Delta=00001_4$ \\
    1 & $\SpecLCE_{4^1}(10131_4, 001\mathbf{0}2_4) \rightarrow \Delta=00011_4$ \\
    1 & $\SpecLCE_{4^1}(102\mathbf{0}1_4, 00112_4) \rightarrow \Delta=00021_4$ \\
    3 & $\SpecLCE_{4^3}(1\mathbf{0}211_4, 0\mathbf{0}122_4) \rightarrow \Delta=01021_4$ \\
    {\bf return} call & $\SpecLCE_{4^4}(11211_4, 01122_4)$ \\
    \end{tabular}

    \caption{An execution of $\ShortLCE_{4^4}(i=(10130)_4, j=(00101)_4)$ (assuming $\LCE(i,j)>4^4$).
    The numbers are given in base-4 representation. Note that there is no $\SpecLCE_{4^2}$ call.}
    \end{figure}

\begin{lemma}
\label{lem:short2}
A sequence of $q$ queries $\ShortLCE_{4^{k_i}}$ for $i\in \{1,\ldots,q\}$ can be answered online in total time 
$\Oh((n+s)\cdot\alpha(n+s,n))=\Oh(n\log^*n+s)$
where $s=\sum_{i=1}^{q}(k_i+1)$.
\end{lemma}
    
\begin{proof}
We calculate $\ShortLCE_{4^k}(i,j)$ using $\Oh(k)$ $\SpecLCE$ queries; see \cref{alg:short_lce_4k}.
We iterate through $k'=0,1,\ldots,k-1$ maintaining $\Delta$ such that $0\le\Delta \le \LCE(i,j)$ and $i+\Delta, j+\Delta\in \S(4^{k'})$.
Before incrementing $k'$, we keep increasing $\Delta$ by $4^{k'}$ until $i+\Delta, j+\Delta\in \S(4^{k'})$
or $\Delta > \LCE(i,j)$.
The latter condition is checked by calling $\SpecLCE_{4^{k'}}(i+\Delta,j+\Delta)$
and terminating if it returns less than $4^{k'}$. 
The while loop iterates at most twice, because $h_{4^{k'+1}}\in \{0,4^{k'},2\cdot 4^{k'} \}$. 
Eventually, we either terminate having found the answer, or we can obtain it with a single call to $\SpecLCE_{4^k}(i+\Delta,j+\Delta)$.

        \begin{algorithm}[ht]
    \caption{$\ShortLCE_{4^k}(i,j)$\label{alg:short_lce_4k}}
    $\ell = \Delta = 0$\;
    \For{$k'=0$ \KwSty{to} $k-1$}{
      \While{$i+\Delta\not\in \S(4^{k'+1})\ \mbox{\bf or}\ j+\Delta\not\in \S(4^{k'+1})$}{
        $\ell = \ell + \SpecLCE_{4^{k'}}(i+\Delta,j+\Delta)$ \Comment*[r]{$i+\Delta,j+\Delta \in \S(4^{k'})$}
        $\Delta = \Delta + 4^{k'}$ \;
        \lIf{$\ell < \Delta$}{%
          \KwRet{$\min(4^k, \ell)$}
        }
      }
    }
    \KwRet{$\min(4^k, \Delta + \SpecLCE_{4^k}(i+\Delta,j+\Delta))$}
      \Comment*[r]{$i+\Delta,j+\Delta \in \S(4^k)$}
    \end{algorithm}

Let us analyze the total time complexity.
Each call to $\ShortLCE_{4^k}$ performs up to $k$ $\SpecLCE_{4^{k'}}$ queries,
but we terminate as soon as we obtain an answer other than~$4^{k'}$.
In \cref{lem:spec2}, the last of these queries contributes $\Oh(k'+1)=\Oh(k+1)$ to $s$,
while the remaining queries contribute one each. 
The total contribution of all $\SpecLCE_{4^{k'}}$ queries called by a single $\ShortLCE_{4^k}$ query 
is therefore $\Oh(k+1)$. Hence, the total running time consumed by all $\SpecLCE_{4^{k'}}$
queries is $\Oh((n+s)\cdot \alpha(n+s,n))$ where $s=\Oh(\sum_{i=1}^q (k_i+1))$. 
It is not hard to see that the remaining time consumed by a single $\ShortLCE_{4^k}$ query is $\Oh(k+1)$.
This is partly because checking whether $i+\Delta$ and $j+\Delta$ belong to $\S(4^{k'+1})$ takes constant time,
since we know that these indices are in $\S(4^{k'})$. 
Over all queries this sums up to $\Oh(s)$, which is dominated by the running time of the $\SpecLCE_{4^{k'}}$ queries.
The $\Oh(n\log^*n + s)$ upper bound follows from \cref{lem:ackermann}.
\end{proof}

\subsection{Final algorithm}

We first modify the implementation details for $\CoarseLCE$ to reduce the preprocessing time. 
\begin{lemma}
\label{lem:fastercoarse}
For $t=\Omega(\log^6 n)$ we can preprocess a string of length $n$ in $\Oh(n\log^*n)$ time, 
so that each $\CoarseLCE_t$ query can be answered in constant time.
\end{lemma}
\begin{proof}
We set $k = \ceil{\frac12 \log t}$ and lexicographically sort all $4^k$-blocks using $\ShortLCE_{4^k}$ queries.
The number of blocks is at most $(\frac{3}{4})^{k}n \leq  \frac{n}{t^{0.5 \log 0.75}}\leq \frac{n}{t^{0.2}}$. By \cref{lem:short2},
the sorting time is:
$$\Oh\left(\frac{n}{t^{0.2}}\log n\log t + n\log^*n\right) = \Oh\left(n \tfrac{\log n \log \log n}{\log^{1.2}n} + n\log^*n\right) = \Oh(n\log^*n).$$
Then we proceed as in the proof of \cref{lem:coarse}.
\end{proof}
By combining \cref{lem:fastercoarse,lem:short2}, we obtain the final theorem.

\begin{theorem}
\label{thm:final2}
A sequence of $q$ LCE queries for a
string over a general ordered alphabet can be executed on-line in total time
$\Oh(q\log \log n + n\log^*n)$ making $\Oh(q+n)$ symbol comparisons.
\end{theorem}

\section{Final remarks}
We gave an $\Oh(n\log\log n)$-time algorithm for answering on-line $\Oh(n)$ LCE queries
for general ordered alphabet.
It is known (see~\cite{Kosolobov2016241}) that the runs of the string can be computed
in $\Oh(T(n))$ time, where $T(n)$ is the time to execute on-line $\Oh(n)$ LCE queries. Hence our algorithm implies the following result:
\begin{corollary}
The runs of a string over general ordered alphabet can
be computed in $\Oh(n\log \log n)$ time.
\end{corollary}

Our algorithm is a major step towards a positive answer for a question posed by Kosolobov~\cite{Kosolobov2016241}, who asked if $\Oh(n)$ time algorithm is possible.

It is also natural to consider general unordered alphabets, that is, strings where
the only allowed operation is checking equality of two characters. 
\begin{theorem}\label{thm:general}
A sequence of $q$ LCE queries for a
string over a general unordered alphabet can be executed in $\Oh(q\log n+ n\log^* n)$ time making $O(n+q)$ symbol equality-tests.
\end{theorem}
\begin{proof}
We can use the faster $\ShortLCE_{4^k}$ algorithm described in~\cref{sec:faster} with $k = \lceil\frac12 \log n\rceil$.
Observe that in this approach  we did not use the order of the characters, and thus it still works for unordered alphabets.
\end{proof}

Note that for unordered alphabets the reduction by Kosolobov~\cite{Kosolobov2016241} (see also~\cite{B14bis}) from computing runs to LCE queries no longer works.
Actually, deciding whether a given string is square-free already requires $\Omega(n\log n)$ comparisons,
as shown by Main and Lorentz~\cite{DBLP:journals/jal/MainL84}.
On the other hand for $\Oh(n)$ $\LCE$ queries $\Oh(n)$ equality tests always suffice.

\bibliographystyle{plainurl}
\bibliography{lce_queries}

\end{document}